\documentclass[conference]{IEEEtran}
\IEEEoverridecommandlockouts

\title{Temperature-Aware Phase-shift Design of LC-RIS for Secure Communication}
\author{
\IEEEauthorblockN{Mohamadreza Delbari$^1$, Bowu Wang$^1$, Nairy Moghadas Gholian$^1$, Arash Asadi$^2$, and Vahid Jamali$^1$}
\IEEEauthorblockA{
$^1$Technical University of Darmstadt (TUD), Darmstadt, Germany, $^2$Delft University of Technology, Delft, Netherlands
\vspace{-3mm}
\thanks{Delbari and Jamali’s work was supported in part by the Deutsche Forschungsgemeinschaft (DFG, German Research Foundation) within the Collaborative Research Center MAKI (SFB 1053, Project-ID 210487104) and in part by the LOEWE initiative (Hesse, Germany) within the emergenCITY center. Moghadas Gholian and Asadi's work was funded by the German Research Foundation (DFG) through the project HyRIS (Grant no. 455077022).}
}
}

\usepackage{amsfonts,amsmath,amsthm,bm}
\usepackage{cite}
\usepackage{amsmath,amssymb,amsfonts}
\usepackage{algorithmic}
\usepackage{graphicx}
\usepackage{textcomp}
\usepackage{xcolor}
\usepackage[T1]{fontenc}
\usepackage{comment}
\usepackage{eucal}
\usepackage{algorithm}
\usepackage{algorithmic}
\usepackage[acronyms,nonumberlist,nopostdot,nomain,nogroupskip]{glossaries}
\usepackage[font=footnotesize]{caption}
\usepackage{subcaption}
\usepackage{booktabs,flushend}

\newtheorem{lem}{Lemma}

\newcommand{\defeq}{\triangleq}
\def\bOmega{\boldsymbol{\Omega}}

\newcommand{\e}{\mathsf{e}}
\newcommand{\jj}{\mathsf{j}}
\newcommand{\Herm}{\mathsf{H}}
\newcommand{\Trans}{\mathsf{T}}
\newcommand{\x}{\mathsf{x}}
\newcommand{\y}{\mathsf{y}}
\newcommand{\z}{\mathsf{z}}
\newcommand{\bA}{\mathbf{A}}
\newcommand{\bx}{\mathbf{x}}

\newcommand{\bs}{\mathbf{s}}
\newcommand{\bS}{\mathbf{S}}

\newcommand{\bH}{\mathbf{H}}
\newcommand{\ba}{\mathbf{a}}

\newcommand{\bu}{\mathbf{u}}

\newcommand{\bh}{\mathbf{h}}

\newcommand{\bg}{\mathbf{g}}

\newcommand{\bq}{\mathbf{q}}
\newcommand{\bp}{\mathbf{p}}

\newcommand{\kk}{\kappa}

\newcommand{\bSigma}{\boldsymbol{\Sigma}}
\newcommand{\bGamma}{\boldsymbol{\Gamma}}
\newcommand{\bmu}{\boldsymbol{\mu}}

\newcommand{\blambda}{\boldsymbol{\lambda}}

\newcommand{\Ex}{\mathbb{E}}

\newcommand{\diag}{\mathrm{diag}}
\newcommand{\real}{\mathrm{Re}}
\newcommand{\imag}{\mathrm{Im}}
\newcommand{\tr}{\mathrm{tr}}
\newcommand{\rank}{\mathrm{rank}}
\newcommand{\dd}{\mathrm{d}}

\newcommand{\tx}{\mathrm{tx}}

\newcommand{\SNR}{\mathrm{SNR}}
\newcommand{\RS}{\mathrm{SR}}
\newcommand{\RIS}{\mathrm{RIS}}

\def\scat{\mathrm{scr}}

\def\bomega{\boldsymbol{\omega}}

\def\bzero{\boldsymbol{0}}
\def\bone{\boldsymbol{1}}
\def\Cset{\mathbb{C}}

\def\Rset{\mathbb{R}}

\def\LOS{\mathrm{LOS}}
\def\nLOS{\mathrm{nLOS}}
\def\tmax{\mathrm{max}}
\def\tx{\mathrm{tx}}
\def\rx{\mathrm{rx}}
\def\BS{\mathrm{BS}}
\def\RIS{\mathrm{RIS}}
\def\NF{\mathrm{NF}}

\def\SNR{\mathrm{SNR}}
\def\eff{\mathrm{eff}}

\def\sCN{\mathcal{CN}}

\def\Pset{\mathcal{P}}
\def\bigO{\mathcal{O}}
\usepackage{xcolor}

\newacronym{RIS}{RIS}{reconfigurable intelligent surface}
\newacronym{QoS}{QoS}{quality of service}
\newacronym{LC}{LC}{liquid crystal}
\newacronym{SNR}{SNR}{signal to noise ratio}
\newacronym{TDMA}{TDMA}{time-division multiple-access}
\newacronym{BS}{BS}{base station}
\newacronym{MU}{MU}{mobile user}
\newacronym{ME}{ME}{mobile eavesdropper}
\newacronym{NF}{NF}{near-field}
\newacronym{Tx}{Tx}{transmitter}
\newacronym{Rx}{Rx}{receiver}
\newacronym{AWGN}{AWGN}{additive white Gaussian noise}
\newacronym{w.r.t.}{w.r.t.}{with respect to}
\newacronym{RDE}{RDE}{Reaction-Diffusion Equation}
\newacronym{PDE}{PDE}{partial differential equation}
\newacronym{UPA}{UPA}{uniform planar array}
\newacronym{AO}{AO}{alternative optimization}
\newacronym{SOCP}{SOCP}{second-order cone programming}
\newacronym{AoD}{AoD}{angle of departure}
\newacronym{LOS}{LOS}{line of sight}
\newacronym{nLOS}{nLOS}{non-LOS}
\newacronym{MIMO}{MIMO}{multiple-input multiple-output}
\newacronym{RS}{SR}{secure rate}
\newacronym{SDP}{SDP}{semi-definite programming}
\newacronym{6G}{6G}{sixth generation}
\newacronym{CSI}{CSI}{channel state information}

\newacronym{PIN}{PIN}{positive-intrinsic-negative}
\newacronym{RF}{RF}{radio frequency}
\newacronym{MEMS}{MEMS}{micro-electro-mechanical system}
\newacronym{mmWave}{mmWave}{millimeter wave}

\begin{document}

\maketitle

\begin{abstract}
\Gls{LC} technology enables low-power and cost-effective solutions for implementing the \gls{RIS}. However, the phase-shift response of \gls{LC}-\gls{RIS}s is temperature-dependent, which, if unaddressed, can degrade the performance. This issue is particularly critical in applications such as secure communications, where variations in phase-shift response may lead to significant information leakage. In this paper, we consider secure communication through an \gls{LC}-\gls{RIS} and developed a temperature-aware algorithm adapting the \gls{RIS} phase shifts to thermal conditions. Our simulation results demonstrate that the proposed algorithm significantly improves the secure data rate compared to scenarios where temperature variations are not accounted for.
\end{abstract}
\glsresetall
\section{Introduction}
\Glspl{RIS} are potential technology of \gls{6G} wireless communications, enabling programmable radio environments \cite{qingqing2019IRS,di2019smart,yu2021smart,najafi2020physics}.
\Gls{LC} technology has been recently studied as a cost-effective and energy-efficient solution for \gls{RIS} implementation, particularly for \gls{mmWave} communication systems \cite{zografopoulos2019liquid,aboagye2022design}. \Gls{LC}s and \gls{LC}-\gls{RIS}s have been investigated in the literature from both experimental and theoretical perspectives  \cite{aboagye2022design,zografopoulos2019liquid,neuder2023compact,jimenez2023reconfigurable,wang2004correlations,Wang2005,delbari2024fast}. For example, Neuder \textit{et al.} \cite{neuder2023compact} demonstrated an experimental design of an \gls{LC}-\gls{RIS}, while Aboagye \textit{et al.} \cite{aboagye2022design} focused on its applications in visible light communication. Additionally, Jim{\'e}nez-S{\'a}ez \textit{et al.} \cite{jimenez2023reconfigurable} provided a comprehensive review of key characteristics of \gls{LC}-\gls{RIS}, including power consumption and cost, and compared these with other related technologies. Building on the works by Wang \textit{et al.} \cite{wang2004correlations,Wang2005}, where the equations for the response time of liquid crystals were derived, the authors in \cite{delbari2024fast} formulated an optimization problem to reduce the switching time response of \gls{LC}-\gls{RIS} systems. These works highlight the growing focus on addressing the efficiency and practicality of \gls{LC} technologies in \gls{RIS}-assisted systems.

\gls{LC}-based \gls{RIS}s are inherently temperature-dependent, as they rely on the mechanical reorientation of \gls{LC} molecules to produce different phase shifts. 
This dependence could lead to significant performance degradation, particularly in the context of physical layer security, due to unintended information leakage. This paper addresses the impact of temperature fluctuations on \gls{LC}-\gls{RIS} for secure communications, focusing on mitigating the effects of temperature-induced information leakage. To the best of the authors' knowledge, the impact of temperature variations on \gls{LC}-\gls{RIS}s has not been investigated in the literature, yet. Our key contributions are as follows:


\begin{figure}
    \centering
    \includegraphics[width=0.4\textwidth]{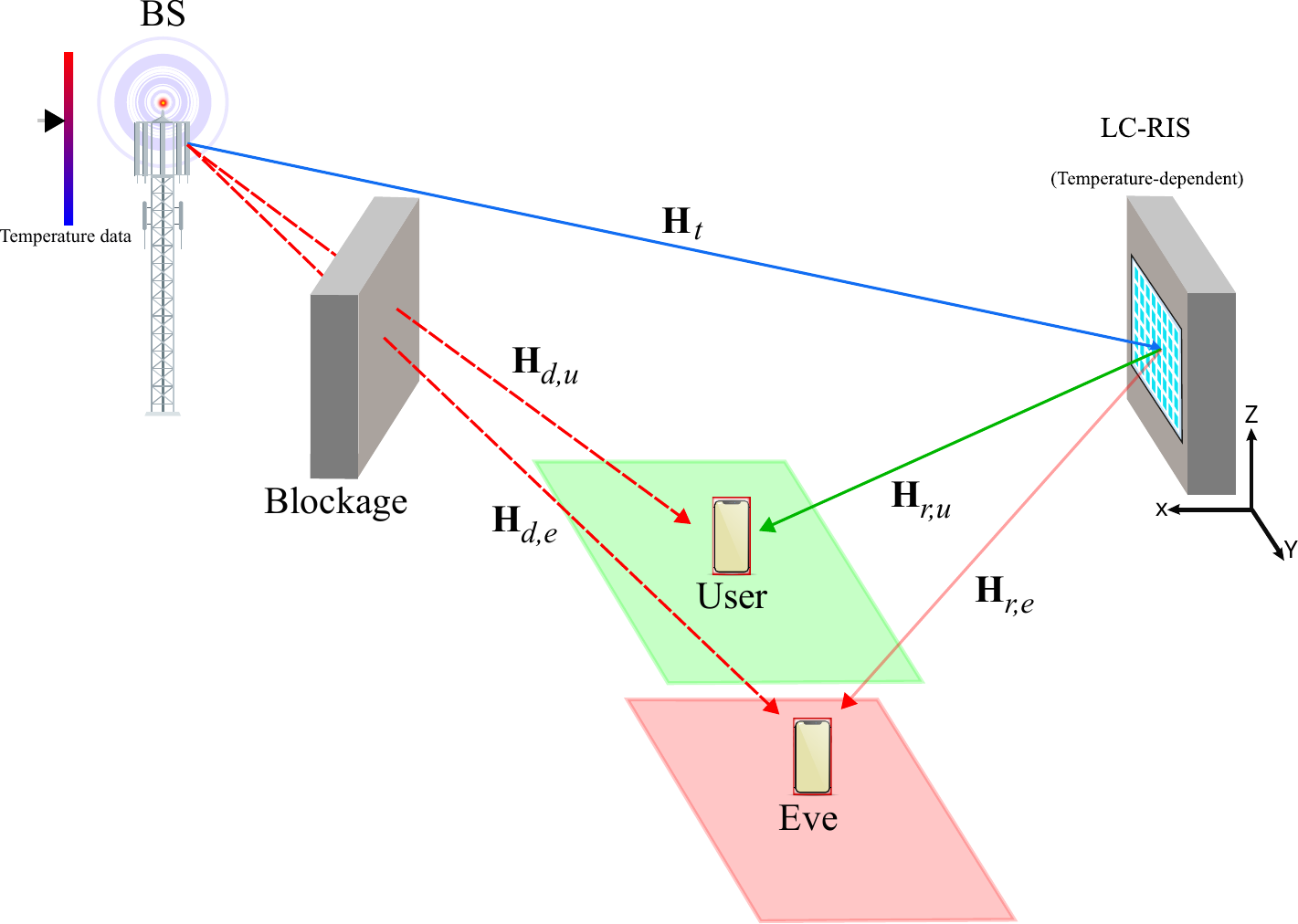}
    \caption{A schematic of a wireless channel model where an \gls{BS} serves a legitimate user via an \gls{RIS} trying to decrease the received signal by an eavesdropper.}
    \label{fig:system model}
    \vspace{-0.5 cm}
\end{figure}

\begin{itemize}
    \item First, we develop a model that accounts for the impact of temperature variations on the phase-shift response of \gls{LC}-\gls{RIS}. We show that temperature fluctuations can influence the range of achievable differential phase shifts $[0-\Delta\omega_{\max}]$, with $\Delta\omega_{\max}$ being less than $2\pi$ in certain cases.
    \item Next, we propose a novel secrecy setup in which we consider a scenario involving an LC-RIS-enhanced system with a legitimate user and a mobile eavesdropper. In this setup, we assume that the exact location of the eavesdropper and its instantaneous \gls{CSI} are unavailable; instead, only the eavesdropper's location with a certain vicinity is known. This is a practical yet challenging assumption, as it reflects real-world scenarios where the RIS must ensure secure communication without precise knowledge of the eavesdropper's location.
    \item Then, we formulate an optimization problem to design phase shift of the \gls{LC}-\gls{RIS} according to the limitation of the phase shift range. Since this problem is non-convex, we propose several reformulations of the problem, which allow us to develop an efficient sub-optimal solution based on rank relaxation and semidefinite programming.
    \item Finally, we evaluate the performance of the proposed algorithm compared to the case \gls{RIS} neglects the impact of the changing temperature. Our simulation results reveal that the secure rate significantly improves when the impact of the temperature variations is accounted for in the \gls{LC}-\gls{RIS} phase-shift design.
\end{itemize}

\textit{Notation:} Bold capital and small letters are used to denote matrices and vectors, respectively.  $(\cdot)^\Trans$, $(\cdot)^\Herm$, $\rank(\cdot)$, and $\tr(\cdot)$ denote the transpose, Hermitian, rank, and trace of a matrix, respectively. Moreover, $\diag(\bA)$ is a vector that contains the main diagonal entries of matrix $\bA$, $\bone_n$ and $\bzero_n$ denote column vectors of size $n$ whose elements are all ones and zeros, respectively. $\|\bA\|_*=\sum_i \sigma_i$, $\|\bA\|_2=\max_i \sigma_i$, $\|\bA\|_F$, and $\blambda_{\max}(\bA)$ denote the respectively nuclear, spectral, and Frobenius norms of a Hermitian matrix $\bA$, and eigenvector associated with the maximum eigenvalue of matrix $\bA$, where $\sigma_i,\,\,\forall i$, are the singular values of $\bA$. Furthermore, $[\bA]_{m,n}$ and $[\ba]_{n}$ denote the element in the $m$th row and $n$th column of matrix $\bA$ and the $n$th entry of vector $\ba$, respectively. $x^+$ denotes as $\max\{x,0\}$. Moreover, $\Rset$ and $\Cset$ represent the sets of real and complex numbers, respectively, $\jj$ is the imaginary unit, and $\Ex\{\cdot\}$  represents expectation. $\mathrm{rand}(N)$ denotes a $N\times1$ vector where each element is generated independently and uniformly from 0 to 1. $\mathcal{CN}(\bmu,\bSigma)$ denotes a complex Gaussian random vector with mean vector $\bmu$ and covariance matrix $\bSigma$. Finally, $\bigO(\cdot)$ represents the big-O notation and $|\Pset|$ is the cardinality of set $\Pset$.

\section{System, LC, and Secrecy rate Model}
\label{section 2}
In this section, we first present the system model for both legitimate user and mobile eavesdropper. Subsequently, we describe the phase-shift model for \gls{LC}-\gls{RIS}. Finally, we introduce the secrecy rate considered in this paper.

\subsection{System and channel models}
\label{system model}
In this paper, we study a narrow-band downlink communication scenario including an \gls{BS} comprising $N_t$ \gls{Tx} antennas, an \gls{RIS} with $N$ \gls{LC}-based unit cells, a single antenna legitimate user, and a single antenna eavesdropper. 
The receive signal at legitimate \gls{MU} and \gls{ME} are
\begin{align}
\label{Eq:system model user}
	y_g = \big(\bh_{d,g}^\Herm + \bh_{r,g}^\Herm \bGamma \bH_t \big) \bx +n_g,\quad g=\{u,e\},
\end{align}
where $\bx\in\Cset^{N_t}$ is the transmit signal vector for the legitimate user, $y_u\in\Cset$ and $y_e\in\Cset$ are the received signal vector at the legitimate user and eavesdropper, respectively, and  $n_u\in\Cset$ ($n_e\in\Cset$) represents the \gls{AWGN} at the legitimate user (eavesdropper), i.e., $n_u,n_e\sim\sCN(0,\sigma_n^2)$, where $\sigma_n^2$ is the noise power. Assuming linear beamforming, the transmit vector $\bx$ can be written as $\bx=\bq s$, where $\bq\in\Cset^{N_t}$ is the beamforming vector on the \gls{BS} and $s\in\Cset$ is the data symbol. Assuming $\Ex\{|s|^2\}=1$, the beamformer satisfies the transmit power constraint $\|\bq\|^2\leq P_t$ with $P_t$ denoting the maximum transmit power. Moreover,  
$\bh_{d,g}\in\Cset^{N_t}, \bH_t\in\Cset^{N\times N_t}$, and $\bh_{r,g}\in\Cset^{N}$ denote the \gls{BS}-\{MU, ME\}, BS-RIS, and RIS-\{MU, ME\} channel matrices, respectively, where $g\in\{u,e\}$. In addition, we assume that the direct channels between the \gls{BS} and both the legitimate mobile user and the mobile eavesdropper are blocked in the coverage area, i.e., $\bh_{d,u}\approx\boldsymbol{0}_{N_t}$ and $\bh_{d,e}\approx\boldsymbol{0}_{N_t}$. This scenario is the primary motivation for utilizing the \gls{RIS} to establish a reliable communication link. Furthermore, $\bGamma\in\Cset^{N\times N}$ is a diagonal matrix with main diagonal entries $[\bGamma]_n=[\bOmega]_n\e^{\jj[\bomega]_n}$ denoting the reflection coefficient applied by the $n$th RIS unit cell comprising phase shift $[\bomega]_n$  and reflection amplitude $[\bOmega]_n$. There are not usually significant amplitude variations on \gls{LC}-\gls{RIS}s within a given narrow frequency band \cite{yang2020design} so we assume $[\bOmega]_n\approx1,\,\forall n$.

We assume that due to the use of higher frequencies, \gls{LOS} links dominate the communication between the \gls{RIS}, \gls{BS}, and the legitimate user. Therefore, we adopt a Rician channel model \cite{jamali2023impact} with a high $K$-factor, indicating the relative power of the \gls{LOS} component compared to the \gls{nLOS} components. This can be applied to each channel component $\bH_t$, $\bh_{r,g}$, and $\bh_{d,g},\, g\in\{u,e\}$ with suitable modifications.

\subsection{LC phase shifter}
\label{sec: LC phase shifter}
In this subsection, we assume a fixed temperature and explain how LC molecules introduce a phase shift into the impinging wave. Based on this discussion, in Section \ref{Theory behind the temperature impact on LC}, we develop a model accounting for the impact of the temperature. LC-RISs control how signals reflect by leveraging the unique electromagnetic properties of LC molecules, which can be reoriented using an electric field ($\vec{E}_{\rm RF}$) \cite{jimenez2023reconfigurable}. This reorientation changes the permittivity of the LC, which in turn affects the phase shift induced by each RIS element. LC molecules have an elongated, rod-like shape, and their electromagnetic properties change depending on whether the electric field is aligned with their major or minor axis. When the electric field is aligned with the major axis, the permittivity is higher, resulting in a greater phase shift, see Fig. \ref{fig:V_phase} for an illustration. By controlling the orientation of these molecules via the applied voltage controlled by the applied electrical field, the RIS can adjust the way signals reflect, allowing it to create a programmable wireless environment.
\begin{figure}[tbp]
	\centering
    \includegraphics[width=0.45\textwidth]{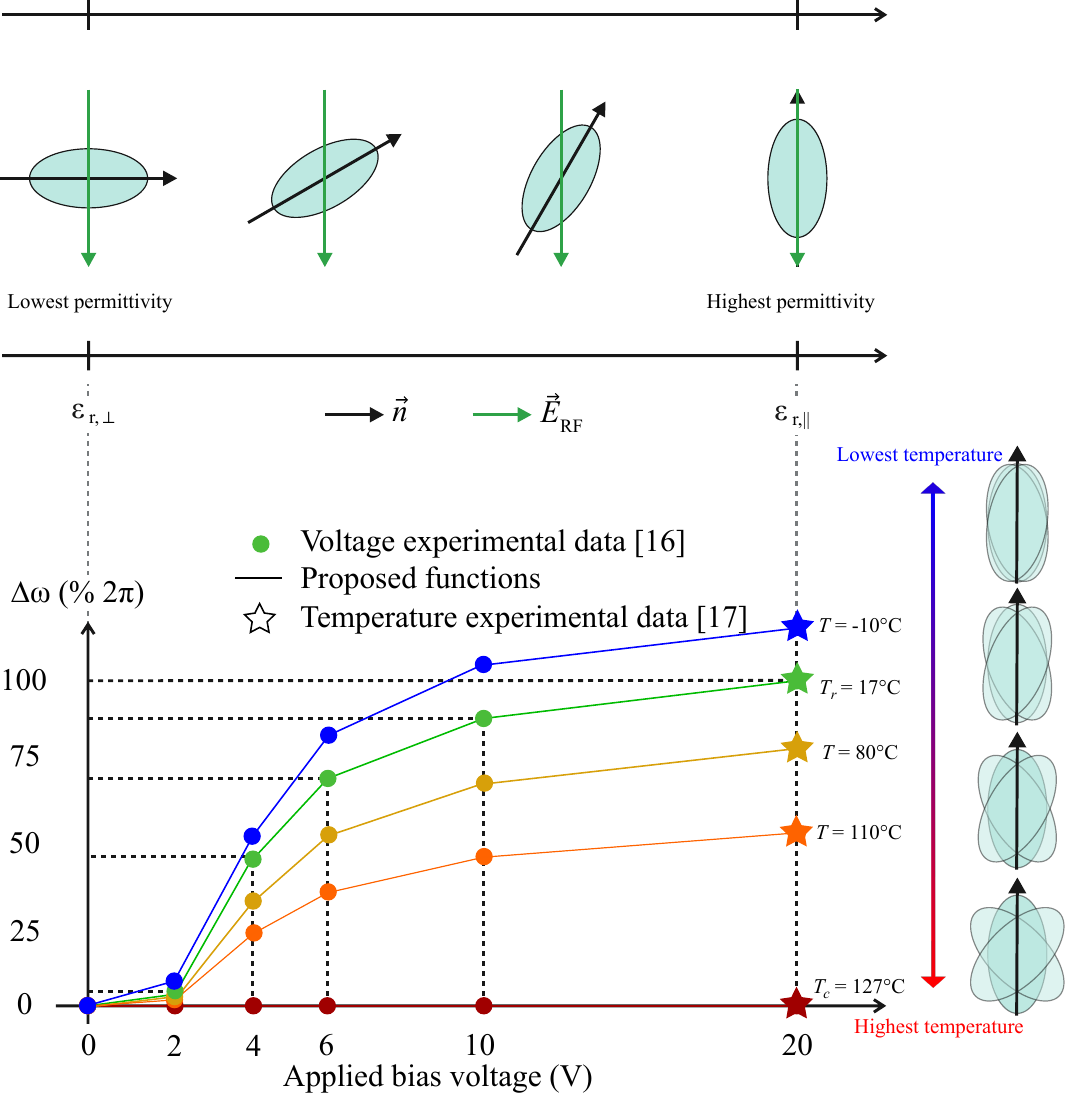}
    \caption{The relationship between phase shift and applied voltage for an \gls{LC} at different temperatures are depicted. The experimental data is taken from \cite{neuder2024architecture,tesmer2021temperature}, a piece-wise linear function is modeled similar to \cite{delbari2024fast}.}
    \vspace{-5mm}
    \label{fig:V_phase}
\end{figure}

The maximum range of phase shift that an \gls{LC} unit can achieve depends on several factors such as the length of the \gls{LC} cell, the difference in
permittivity between different molecular orientations, and the operating frequency, i.e.,
\begin{equation}
    \Delta\omega_{\max}=2\pi l(\sqrt{\varepsilon_{r,\parallel}}-\sqrt{\varepsilon_{r,\perp}})\frac{f}{c},
    \label{eq:omega epsilon}
\end{equation}
where $l$, $f$, and $c$ are the phase shifter length, frequency, and speed of the light in vacuum, respectively. $\varepsilon_{r,\parallel}$ and $\varepsilon_{r,\perp}$ are the maximum and the minimum relative permittivity, respectively. They correspond to when an electrical field is aligned or perpendicular to the vector $\vec n$ in Fig. \ref{fig:V_phase}. 
The value of $\varepsilon_{r,\parallel}$ and $\varepsilon_{r,\perp}$ are scaled with temperature, which leads to a temperature-dependent phase-shift response for \gls{LC}-\gls{RIS}s, as will be discussed in detail in Section \ref{Temperature-Aware LC-RIS Phase-shift Design}.
\subsection{Secrecy rate}
The secrecy rate is a crucial metric used to quantify the performance of a physical-layer secure communication system. It is defined as the difference between the achievable rates of the legitimate user and the eavesdropper \cite{Cheng2023secure}: 
\begin{IEEEeqnarray}{ll}
\label{eq: SNR}
    \RS=[\log(1+\SNR_u)-\log(1+\SNR_e)]^+,\\
    \SNR_u=\frac{|(\bh_u^\eff)^\Herm\bq|^2}{\sigma^2_n},\\
    \SNR_e=\frac{|(\bh_e^\eff)^\Herm\bq|^2}{\sigma^2_n},
\end{IEEEeqnarray}
where $(\bh_g^\eff)^\Herm=\bh_{d,g}^\Herm + \bh_{r,g}^\Herm \bGamma \bH_t,\,g=\{u,e\}$ is the end-to-end channel from the \gls{BS} to the legitimate user and eavesdropper.

In this paper, our goal is to maximize the secrecy rate by optimizing the \gls{RIS} phase shifts and \gls{BS} beamformer, thus the system can create a favorable communication environment for the legitimate user while suppressing the signal strength in the area of the eavesdropper. We design the \gls{RIS} phase-shift based on the approximate locations of both the legitimate user and the eavesdropper, with the following benefits:
\begin{itemize}
    \item Unlike most literature that assumes the eavesdropper's \gls{CSI} is fully known, we assume only that the eavesdropper’s location lies within a specified area, $\bp_e\in\Pset_e$. The size of this area, $|\Pset_e|$, accounts for potential errors in estimating this location.
    \item Similarly, for the legitimate user, we assume only that their location is within a specific area, $\bp_u\in\Pset_u$. The area size, $|\Pset_e|$, can be adjusted not only to account for location estimation error but also to reduce reconfiguration overhead by decreasing the frequency of required \gls{RIS} adjustments.
\end{itemize}
The \gls{RIS} phase shifts will be designed based on the {LOS} paths to optimize the communication, as it offers the most significant contribution to signal strength at \gls{mmWave} bands, where \gls{LC}-\gls{RIS}s are expected to be adopted.

\section{Temperature-Aware LC-RIS Phase-shift Design}
\label{Temperature-Aware LC-RIS Phase-shift Design}
In this section, we start by introducing a model to describe the impact of the temperature on the \gls{LC}-\gls{RIS} phase shifts. Then, we formulate an optimization problem for phase-shift design aimed at maximizing the secure rate for a given temperature.

\subsection{Theory behind the temperature impact on LC}
\label{Theory behind the temperature impact on LC}


In this section, we analyze how temperature variations affect the maximum phase shift generated by an individual \gls{LC} cell. Let us define $\Delta n\triangleq(\sqrt{\varepsilon_{r,\parallel}}-\sqrt{\varepsilon_{r,\perp}})$ in \eqref{eq:omega epsilon}, which is known in the literature as birefringence \cite{Wang2005}. With changing the temperature, $\Delta n$ changes leading to a different phase shift. The dependency of $\Delta n$ on temperature is characterized by the Haller approximation \cite{haller1975thermodynamic}:
\begin{equation}
\label{eq: delta n}
    \Delta n=\Delta n_0\Big(1-\frac{T}{T_c}\Big)^\beta,
\end{equation}
where $\Delta n_0$ is birefringence at the reference temperature, denoted by $T_r$, $T$ is the current temperature and $T_c$ is the clearing temperature of the \gls{LC}, which is defined as the temperature at which the \gls{LC} becomes isotropic, i.e., \gls{LC} molecule becomes completely disordered. Moreover, $\beta$ is a constant which typically is around $0.2-0.25$ for many \gls{LC} materials \cite{Wang2005}. 
Let us assume at $T_r$, \gls{LC} cells produce $2\pi$ differential phase shift. Exploiting \eqref{eq:omega epsilon} and \eqref{eq: delta n}, it can be obtained:
\begin{equation}
    \Delta\omega_{\max}(T)=2\pi l\Delta n_0\Big(1-\frac{T}{T_c}\Big)^\beta\frac{f}{c},
    \label{eq:omega temperature}
\end{equation}
 where $\Delta\omega_{\max}(T_r)=2\pi$. Equivalently, \eqref{eq:omega temperature} can be rewritten as

 \begin{equation}
    \Delta\omega_{\max}(T)=2\pi\Big(\frac{T_c-T}{T_c-T_r}\Big)^\beta.
    \label{eq:omega temperature final}
 \end{equation}
Therefore, when the temperature decreases from $T_r$, the maximum achievable phase shift increases. In this case, $\Delta\omega_{\max}$ exceeds $2\pi$, allowing the phase shifts to be adjusted back to those corresponding to $T_r$ by tuning the voltages. However, with increasing temperature, the maximum achievable phase shift decreases. This case will be investigated in the following section.

\subsection{Temperature-aware phase-shift design}
Next, assume the setup in Fig. \ref{fig:system model} where \gls{RIS} wishes to maximize the secure rate defined in \eqref{eq: SNR} for the user in the worst case. In other words, the secure rate (SR) must be maximized regardless of the exact legitimate and eavesdropper's locations while they are inside of their own defined areas. We formulate the following problem formulation based on the only \gls{LOS} links:
\begin{subequations}
\label{eq:optimization 1}
\begin{align}
    \text {P1:}\quad&~\underset{\bomega,\bq,\alpha}{\max}~\alpha
    \\&~\text {s.t.} ~~\RS\geq \alpha,\, \forall \bp_u\in\Pset_u,\,\forall\bp_e\in\Pset_e
    \\&\quad\hphantom {\text {s.t.} } 0\leq [\bomega]_n < \omega_\tmax, \forall n,
    \\&\quad\hphantom {\text {s.t.} } \|\bq\|_2^2\leq P_t.
\end{align}
\end{subequations}
Here, (\ref{eq:optimization 1}b) is the secure rate constraint defined in \eqref{eq: SNR}, (\ref{eq:optimization 1}c) is a constraint over phase shifts which is rooted in the temperature as we discussed in \ref{Theory behind the temperature impact on LC}, and (\ref{eq:optimization 1}d) forces the transmit power. In addition, $\alpha$ describes the worst case for secure rate and we try to maximize this parameter. The problem P1 is non-convex due to the non-convexity in the constraint (\ref{eq:optimization 1}b). Two vector variables $\bq$ and $\bomega$ are also coupled in this constraint making it more challenging to derive a global solution for this problem. In the following, we divide this problem into two sub-problems and try to maximize the cost function at each step via \gls{AO}.\\

\textbf{Beamformer design:} In this step, we assume $\bomega$ is fixed and the only variable of the problem is $\bq$. In addition, since the LOS link is the dominant path at higher frequencies, we design the beamformer based on the \gls{LOS} link. This assumption is often valid especially because both the BS and RIS are positioned at elevated locations above the ground. Based on this assumption, $\bH_t$ can be decomposed as:
\begin{equation}
\label{eq: H_t}
    \bH_t=c_0\ba_\text{RIS}(\bp_\BS)\ba^\Herm_\BS(\bp_\RIS),
\end{equation}
where $\ba_\RIS(\cdot)$ and $\ba^\Herm_\BS(\cdot)$ are the steering vectors at the \gls{RIS} and \gls{BS}, respectively, and $\|\ba_\RIS(\cdot)\|=\|\ba^\Herm_\BS(\cdot)\|=1$. Due to the large \gls{LC}-\gls{RIS}, we assume \gls{NF} model for $\ba_\RIS(\cdot)$ \cite{delbari2024nearfield}. Moreover, $\bp_\BS$ and $\bp_\RIS$ are the locations of \gls{BS} and \gls{RIS}, respectively. $c_0$ denotes the channel attenuation factor of the \gls{LOS} link. By assuming a fixed $\bomega$, the problem P1 reduces to the following sub-problem:
\begin{subequations}
\label{eq:optimization 2}
\begin{align}
    \text {P2:}\quad&~\underset{\bq,\alpha}{\max}~\alpha
    \\&~\text {s.t.} ~~\RS\geq \alpha,\, \forall \bp_u\in\Pset_u,\,\forall\bp_e\in\Pset_e,
    \\&\quad\hphantom {\text {s.t.} } \|\bq\|_2^2\leq P_t.
\end{align}
\end{subequations}
The optimal beamformer is obtained using the following lemma.
\begin{lem}
\label{lemma beamforming}
    Assuming a dominant \gls{LOS} link, fixed \gls{RIS} phase shifts, and blocked direct channels for both the legitimate user and the eavesdropper, the optimal beamformer for P2 is given by $\bq=\sqrt{P_t}\ba_\BS(\bp_\RIS)$.
\end{lem}
\begin{proof}
 Substituting $\bH_t$ from \eqref{eq: H_t} into the secure rate definition \eqref{eq: SNR} and neglecting the $\bh_{d,g},\,g=\{u,e\}$ due to blockage yields:
\begin{equation}
    \RS=\Big[\log\big(\frac{1+\overbrace{|\bh_{r,u}\bGamma c_0\ba_\RIS(\bp_\BS)|^2/\sigma^2_n}^{\zeta_u(\bp_u)}\times   \overbrace{|\ba^\Herm_\BS(\bp_\RIS)\bq|^2}^{\mu}}{1+\underbrace{|\bh_{r,e}\bGamma c_0\ba_\RIS(\bp_\BS)|^2/\sigma^2_n}_{\zeta_e(\bp_e)}\times   \underbrace{|\ba^\Herm_\BS(\bp_\RIS)\bq|^2}_{\mu}}\big)\Big]^+,
\end{equation}
where $\zeta_u(\bp_u)$ and $\zeta_e(\bp_e)$ are fixed due to the given RIS phase shifts and $\mu\in\Rset$ is a controllable scalar in terms of $\bq$. When $\zeta_u(\bp_u)>\zeta_e(\bp_e),\,\forall \bp_u\in\Pset_u, \forall \bp_e\in\Pset_e$, $\RS$ is monotonically increasing in $\mu$. The maximum $\mu$ is $P_t$ when we set $\bq=\sqrt{P_t}\ba_{\BS}(\bp_{\RIS})$ \cite{tse2005fundamentals}. Otherwise, when there exists at least one pair $(\bp_u,\bp_e)$, $\bp_u\in\Pset_u$ and $\bp_e\in\Pset_e$, such that $\zeta_u(\bp_u)\leq\zeta_e(\bp_e)$, then we have $\RS=0$ independent of $\mu$. Therefore, without loss of generality, we can choose $\bq=\sqrt{P_t}\ba_{\BS}(\bp_{\RIS})$ which concludes the proof.
\end{proof}

\textbf{RIS design:} Assuming a fixed beamformer, we maximize the secrecy rate by optimizing the \gls{RIS} phase shifts. The phase shift configuration subproblem is given by:
\begin{subequations}
\label{eq:optimization 3}
\begin{align}
    \text {P3:}\quad&~\underset{\bomega,\alpha}{\max}~\alpha
    \\&~\text {s.t.} ~~\text{C1:~}\RS\geq \alpha,\, \forall \bp_u\in\Pset_u,\,\forall\bp_e\in\Pset_e
    \\&\quad\hphantom {\text {s.t.} } \text{C2:~} 0\leq [\bomega]_n < \omega_\tmax, \forall n.
\end{align}
\end{subequations}
Problem P3 is non-convex due to the non-convexity of constraint (\ref{eq:optimization 3}b) over $\omega_n,\,\forall n$. Without loss of generality, we drop $[\cdot]^+$ from $\RS$ in \eqref{eq: SNR} and aim to maximize secure rate by P3. If $\RS>0$ for the obtained solution, then dropping $[\cdot]^+$ has no impact. Otherwise, if $\RS<0$, it implies that the secure rate is zero. By a new definition $\alpha\defeq\log(\gamma)$, constraint (\ref{eq:optimization 3}b) can be rewritten as follows:
\begin{equation}
    \frac{1+\SNR_u}{1+\SNR_e}\geq\gamma\Rightarrow (1-\gamma)+(\SNR_u-\gamma\SNR_e)\geq0.
\end{equation}
 Here, we can also decompose each $\SNR_u$ and $\SNR_e$ in terms of a vector including exponential of \gls{RIS} phase shifts $\bs\defeq[\e^{\jj[\bomega]_1}, \cdots, \e^{\jj[\bomega]_N}]^\Trans$. With this assumption, we have:
\begin{subequations}
    \label{eq: SNR in term of s}
    \begin{align}
        \SNR_u=&\bs^\Herm\bA_u\bs,\\
        \SNR_e=&\bs^\Herm\bA_e\bs,
    \end{align}
\end{subequations}
where $\bA_g=\frac{\diag(\bh_{r,g}^\Herm)\bH_t\bq\bq^\Herm\bH_t^\Herm\diag(\bh_{r,g})}{\sigma_n^2},\,g=\{u,e\}$. Hence, constraint C1 in P3 can be reformulated as
\begin{equation}
    \widehat{\text{C1}}: (1-\gamma)+\bs^\Herm(\bA_u(\bp_u)-\gamma\bA_e(\bp_e))\bs\geq0,\, \forall \bp_u\in\Pset_u,\,\forall\bp_e\in\Pset_e.
\end{equation}
After applying these reformulations in P3 and because logarithm is a monotone-increasing function, the problem P3 can change to $\widehat{\text{P3}}$ in the following:
\begin{subequations}
\label{eq:optimization 3-1}
\begin{align}
\widehat{\text{P3}}:\quad&~\underset{\bomega,\gamma}{\max}~\gamma
    \\&~\text {s.t.} ~~\widehat{\text{C1}},\text{  C2}.
\end{align}
\end{subequations}
 To tackle the non-convexity of the problem, we transform $\widehat{\text{P3}}$ into a \gls{SDP} problem by defining $\bS\defeq\bs\bs^\Herm$, and $\bA(\gamma)\defeq\bA_u(\bp_u)-\gamma\bA_e(\bp_e)$ where $\bA(\gamma)$ is a function of $\bp_u,\,\bp_e,$ and $\gamma$:
\begin{subequations}
\label{eq:optimization 4}
\begin{align}
    \text {P4:}&~\underset{\bS,\gamma}{\max}~\gamma
    \\&~\text {s.t.} ~~\widehat{\widehat{\text{C1}}}:\tr(\bA(\gamma)\bS)\geq\gamma-1, \forall \bp_u\in\Pset_u,\,\forall\bp_e\in\Pset_e,
    \\&\quad\hphantom {\text {s.t.} } \text{C2: } 0\leq [\bomega]_n < \omega_\tmax, \forall n,
    \\&\quad\hphantom {\text {s.t.} }\text{C3: } \bS\succeq 0,\text{C4: } \rank(\bS)=1, \text{C5: }\diag(\bS)=\bone_N.
\end{align}
\end{subequations}
This problem is still non-convex due to the non-convexity in C2 and C4 in terms of $\bS$, and being coupled $\gamma$ and $\bS$ in $\widehat{\widehat{\text{C1}}}$. We resolve the non-convexity of each one in the following.
\subsubsection{Rank constraint C4} To tackle this issue, we adopt the penalty method exploited in \cite{Yu2020power,ghanem2022optimization,delbari2024far}.
The basic idea is to replace the rank constraint with inequality $\|\bS\|_*-\|\bS\|_2 \leq 0$, which holds only if $\bS$ has rank smaller or equal to one. While the new constraint is still non-convex, one can apply the first-order Taylor approximation to make it convex. Let $\bS^{(i)}$ denotes the value of matrix $\bS$ in the $i$th iteration. According to the first-order Taylor approximation:
\begin{equation}
\label{eq: taylor approximation}
    \|\bS\|_2\geq\|\bS^{(i)}\|_2+\tr\big(\blambda_{\max}(\bS^{(i)})\blambda_{\max}^\Herm(\bS^{(i)})(\bS-\bS^{(i)})\big),
\end{equation}
By adopting the penalty method \cite{Yu2020robust} and applying \eqref{eq: taylor approximation} into the cost function of P4, we have
\begin{subequations}
\label{eq:optimization 5}
\begin{align}
    \text {P5:}\quad&~\underset{\bS,\gamma}{\max}~\gamma-\eta^{(i)}\Big(\|\bS\|_*-\|\bS^{(i)}\|_2-\tr\big(\blambda_{\max}(\bS^{(i)})\nonumber
    \\&\quad\quad\quad\times\blambda_{\max}^\Herm(\bS^{(i)})(\bS-\bS^{(i)})\big)\Big)
    \\&~\text {s.t.} ~~\widehat{\widehat{\text{C1}}}, \text{C2, C3, C5}.
\end{align}
\end{subequations}
Here, $\eta^{(i)}$ is the penalty factor at iteration $i$ which increases gradually. By selecting a sufficiently large $\eta$, problems P4 and P5 become equivalent. In the following, we will address the non-convexity of C2 \gls{w.r.t.} $\bS$.
\subsubsection{C2: $0\leq [\bomega]_n < \omega_\tmax,\,\forall n$} While constraint C2 is linear in $[\omega]_n, \forall n$, it is highly non-convex in the new variable $\bS$. To resolve this issue, we extract features of $\bS$ that are informative about $\omega_{\max}$ and can be used to enforce C2. We first prove two lemmas and based on them, we reformulate this constraint to a non-equality in terms of $\bS$ satisfying C2. 
\begin{lem}
\label{lemma integral}
    Let us assume $\omega_n$ is uniformly distributed in interval $[0,\omega_{\max}]$. For sufficiently large $N$, matrix $\bS$ that satisfies the constraints C2-C5, also satisfies:
    \begin{equation}
    \sum_{i=1}^N [\bS]_{n,i}=\frac{\e^{\jj[\bomega]_n}N(1-\e^{-\jj\omega_{\max}})}{{\jj\omega_{\max}}},\forall n.
    \end{equation}
\end{lem}
\begin{proof}
    Recalling $\bS=\bs\bs^\Herm$ with $\bs=[\e^{\jj[\bomega]_1},\,\cdots,\,\e^{\jj[\bomega]_N}]$, we have
    \begin{align}
        \sum_{i=1}^N [\bS]_{n,i}&=\e^{\jj[\bomega]_n}\sum_{i=1}^N \e^{-\jj[\bomega]_i}=\e^{\jj[\bomega]_n}N\sum_{i=1}^N\frac{\e^{-\jj[\bomega]_i}}{N}\\
        &\stackrel{(a)}{=}\frac{\e^{\jj[\bomega]_n}N}{\omega_{\max}} \int_{0}^{\omega_{\max}}\e^{-\jj\omega} \dd\omega=\frac{\e^{\jj[\bomega]_n}N(1-\e^{-\jj\omega_{\max}})}{\jj\omega_{\max}},
    \end{align}
    where $(a)$ holds based on law of large numbers \cite{papoulis2002probability}.
\end{proof}

\begin{lem}
\label{lemma constraint C2}
    The constraint $0\leq[\bomega]_n\leq\omega_\tmax,$
    , where $\pi<\omega_{\max}<2\pi$, is equivalent to
     $\real([\bs]_n)+\tan(\frac{\omega_{\max}}{2})\imag([\bs]_n)\leq1$,
     where $[\bs]_n=\e^{\jj[\bomega]_n}$ and $\real(\cdot)$ and $\imag(\cdot)$ denote the real and imaginary parts, respectively.
\end{lem}
\begin{proof}
    We start by noting that since $[\bs]_n=\e^{\jj[\bomega]_n}$, we can express $\real([\bs]_n)+\tan(\frac{\omega_{\max}}{2})\imag([\bs]_n)$ as $\cos([\bomega]_n)+\tan(\frac{\omega_{\max}}{2})\sin([\bomega]_n)$. By substituting this into the inequality, it yields that
    \begin{align}
        \tan(\frac{\omega_{\max}}{2})\sin([\bomega]_n)\leq1-\cos([\bomega]_n).
    \end{align}
    Using the known trigonometric equations $\sin([\bomega]_n)=2\sin(\frac{[\bomega]_n}{2})\cos(\frac{[\bomega]_n}{2})$ and $1-\cos([\bomega]_n)=2\sin^2(\frac{[\bomega]_n}{2})$, we can transform the inequality into:
    \begin{align}
        2\tan(\frac{\omega_{\max}}{2})\sin(\frac{[\bomega]_n}{2})\cos(\frac{[\bomega]_n}{2})\leq2\sin^2(\frac{[\bomega]_n}{2}).
    \end{align}
    We can divide out $\sin(\frac{[\bomega]_n}{2})$ from both sides of the inequality since $\sin(\frac{[\bomega]_n}{2})\geq0$ for $0\leq[\bomega]_n\leq2\pi$. Consequently, the inequality simplifies to
    \begin{align}
        \tan(\frac{\omega_{\max}}{2})\cos(\frac{[\bomega]_n}{2})\leq\sin(\frac{[\bomega]_n}{2}).
    \end{align}
    This is equivalent to
    \begin{align}
       \begin{cases}
            \textbf{Case 1:} \tan(\frac{\omega_{\max}}{2})\leq\tan(\frac{[\bomega]_n}{2}),\,\,0<[\bomega]_n<\pi,\\
            \textbf{Case 2:} \tan(\frac{\omega_{\max}}{2})\geq\tan(\frac{[\bomega]_n}{2}),\,\,\pi<[\bomega]_n<2\pi.
        \end{cases}
    \end{align}
    Case 1 always holds because $\tan(\frac{\omega_{\max}}{2})<0$ when $\pi<\omega_{\max}<2\pi$ while $\tan(\frac{[\bomega]_n}{2})>0$. Case 2 is satisfied as long as $[\bomega]_n\leq\omega_{\max}$, thus concluding the proof.
\end{proof}
Based on Lemma~\ref{lemma integral} and Lemma~\ref{lemma constraint C2}, we can conclude that the constraint C2 can be written as
\begin{equation}
    \widehat{\text{C2}}:\real(\zeta\sum_{i=1}^N [\bS]_{n,i})+\tan(\frac{\omega_{\max}}{2})\imag(\zeta\sum_{i=1}^N [\bS]_{n,i})\leq1, \forall n,
\end{equation}
where $\zeta\defeq\frac{\jj\omega_{\max}}{N(1-\e^{-\jj\omega_{\max}})}$. Unlike C2, constraint $\widehat{\text{C2}}$ is convex in $\bS$. These two constraints are equivalent under three key assumptions: $(i)$ a sufficiently large number of RIS elements ($N$), $(ii)$ a uniform distribution of phase shifts in the interval $[0,\,\omega_{\max}]$, and $(iii)$ $\pi<\omega_{\max}<2\pi$. From our observations, $N$ greater than $100$ is sufficient for this approximation. While a uniform distribution is not generally guaranteed, the individual phase-shifts tend to be rather random in \gls{NF} regime when \gls{RIS} serves an area, making the assumption reasonable. Regarding the third assumption, experimental results confirm that even at higher temperatures typical in outdoor scenarios, $\omega_{\max}$ does not drop below $\pi$ \cite{tesmer2021temperature}.

Therefore, problem P5 is transformed into the following optimization problem.
\begin{subequations}
\label{eq:optimization 6}
\begin{align}
    \text {P6:}\quad&~\underset{\bS,\gamma}{\max}~\gamma-\eta^{(i)}\Big(\|\bS\|_*-\|\bS^{(i)}\|_2-\tr\big(\blambda_{\max}(\bS^{(i)})\nonumber
    \\&\quad\quad\quad\times\blambda_{\max}^\Herm(\bS^{(i)})(\bS-\bS^{(i)})\big)\Big)
    \\&~\text {s.t.} ~~\widehat{\widehat{\text{C1}}}, \widehat{\text{C2}},\text{ C3, C5}.
\end{align}
\end{subequations}

\subsubsection{Coupled $\gamma$ and $\bS$ in $\widehat{\widehat{\text{C1}}}$} To address the coupling of $\bS$ and $\gamma$ in $\widehat{\widehat{\text{C1}}}$, we employ \gls{AO}, where one variable is fixed while the other is optimized. On one hand, when $\gamma$ is fixed, the optimization problem P6 becomes convex in terms of the matrix $\bS$ because the objective function is concave, and the constraints define a convex set. Therefore, it can be efficiently solved using standard convex optimization solvers such as CVX \cite{cvx}. On the other hand, when the matrix $\bS$ is fixed, the problem P6 simplifies to the following optimization problem:
\begin{subequations}
\label{eq:optimization 6 gamma}
\begin{align}
    \widehat{\text {P6:}}\quad&~\underset{\gamma}{\max}~\gamma
    \\&~\text {s.t.} ~~\widehat{\widehat{\text{C1}}}.
\end{align}
\end{subequations}
The problem $\widehat{\text {P6}}$ is linear in terms of $\gamma$ and its solution is given by:
\begin{equation}
\label{eq: best gamma}
    \gamma=\underset{\forall \bp_u\in\Pset_u,\,\forall\bp_e\in\Pset_e}{\min}~\frac{\tr(\bA_u(\bp_u)\bS)+1}{\tr(\bA_e(\bp_e)\bS)+1}.
\end{equation}
\begin{algorithm}[t]
\caption{Proposed Algorithm for Problem P6}\label{alg:cap}
\begin{algorithmic}[1]
\STATE \textbf{Initialize:} $\omega_{\max}$, $\bs^{(0)}=\e^{\jj\omega_{\max}\times\mathrm{rand}(N)},\bS^{(0)}=\bs^{(0)}{\bs^{(0)}}^\Herm$.
\WHILE{$|\log_2(\gamma^{(j)})-\log_2(\gamma^{(j-1)})|\geq\epsilon_2$ and $j\leq J_{\max}$}
    \WHILE{$\|\bS^{(i)}-\bS^{(i-1)}\|_F^2\geq\epsilon_1$ and $i\leq I_{\tmax}$}
    \STATE Solve convex P6 for given $\bS^{(i-1)}$ and $\gamma$, and store the intermediate solution $\bS$.
    \STATE Set $i = i + 1$ and update $\bS^{(i)}=\bS$ and $\eta^{(i)} =5\eta^{(i-1)}$.
    \ENDWHILE
    \STATE Calculate new $\gamma$ according to the \eqref{eq: best gamma} and set $j = j + 1$.
\ENDWHILE
\end{algorithmic}
\end{algorithm} 
\textbf{Algorithm and complexity analysis:} The proposed algorithm is summarized in Algorithm \ref{alg:cap}. At each iteration, the most computational step is the computation of the nuclear norm in line 4, which has a complexity of $\bigO(N^3)$. The number of different constraints generated by $\widetilde{\text{C1}}$ is the bottleneck and proportional to $|\Pset_u||\Pset_e|$. Thus, the complexity of the Algorithm \ref{alg:cap} in total is $\bigO(I_{\max}J_{\max}|\Pset_u||\Pset_e|N^3)$.

\section{Performance Comparison}
\subsection{Simulation Setup}
We employ the simulation configuration for coverage extension presented in Fig. \ref{fig:system model}, where the \gls{RIS} center is the origin of the Cartesian coordinate system, i.e., $[0,0,0]~\text{m}$. We assume there is a legitimate user with a fixed area in $\Pset_u\in\{(\x,\y,\z):4.5~\text{m}\leq\x\leq 5.5~\text{m}, -0.5~\text{m}\leq\y\leq 0.5~\text{m}, \z=-5\}$. The \gls{BS} comprises a $4\times4=16$ \gls{UPA} located at $[30,0,5]~\text{m}$. The \gls{RIS} is a \gls{UPA} consisting of $N_\y\times N_\z=20\times10$ elements aligned to the $\y$ and $\z$ axes, respectively. The element space for both the \gls{BS} and \gls{RIS} is half of the wavelength. The noise variance is computed as $\sigma_n^2=WN_0N_{\rm f}$ with $N_0=-174$~dBm/Hz, $W=20$~MHz, and $N_{\rm f}=6$~dB. We assume $28$~GHz carrier frequency, and $\rho(d_0/d)^\sigma$ pathloss model where $\rho=-61$~dB at $d_0=1$~m. Moreover, we adopt $\sigma = (2,2,2)$ and  $K=(0,10,10)$ for the \gls{BS}-\gls{MU}, \gls{BS}-\gls{RIS}, and \gls{RIS}-\gls{MU} channels, respectively. 

The analysis considers two scenarios based on the location of eavesdropper \gls{w.r.t.} to the legitimate user area. These scenarios are critical because the eavesdropper's location affects how the \gls{RIS} can optimize the secure rate communication.
\begin{itemize}
    \item \textbf{Scenario 1: Horizontal Case (H):} Here, the eavesdropper is positioned horizontally away from the legitimate user. In this case, the eavesdropper area is at a \underline{different y-coordinate} but the \underline{same x-coordinate} as the legitimate user. This implies that the eavesdropper is positioned to the left or right of the legitimate user area.
    \item \textbf{Scenario 2: Vertical Case (V):} In this case, the eavesdropper is positioned vertically away from the legitimate user. The eavesdropper area has the \underline{same y-coordinate} as the user but is at a \underline{different x-coordinate}, meaning eavesdropper is positioned above or below the legitimate user area.
\end{itemize}
Furthermore, two approaches are considered in managing the \gls{RIS} phase shifts: \textbf{Neglecting (N)} temperature changes, where \underline{no adjustment} is made, and \textbf{Optimizing (O)} the phase shifts to \underline{account for} temperature variations to ensure secure communication. The other parameters used in the simulations are as follows: $\beta=0.25$, $T_c=127~^{\circ}$C, $T_r=17~^{\circ}$C, $P_t=40~$dBm, $T=57~^{\circ}$C, $\eta^{(0)}=0.01$, $I_{\max}=12$, $J_{\max}=4$, $\epsilon_1=0.01$, $\epsilon_2=0.1$, and $\gamma^{(0)}=10^3$ are assumed.
 \begin{figure}[t]
    \centering
    \includegraphics[width=0.4\textwidth]{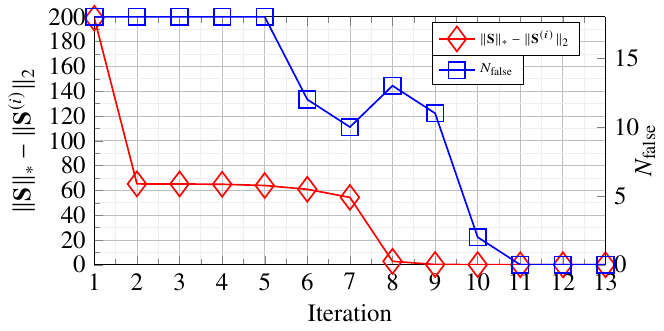}
    \caption{The value of $\|\bS\|_*-\|\bS^{(i)}\|_2$ and $N_\mathrm{false}$ at each iteration}
    \label{fig:convergence}
    \vspace{-5mm}
\end{figure}

\subsection{Simulation Result}
First, we study the convergence behavior of the Algorithm \ref{alg:cap} in Fig.~\ref{fig:convergence} for one example. In this figure, $N_\mathrm{false}$ is the number of the elements whose phase shifts are larger than $\omega_{\max}$ and violating C2. As the number of iterations increases, $\|\bS\|_* - \|\bS^{(i)}\|_2$ and $N_\mathrm{false}$ approach to zero ensuring that the rank-one constraint for $\bS$ and constraint C2 are met.

Fig. \ref{fig: heat map} shows the received signal power at various locations when the eavesdropper is positioned either to the left (horizontal) or below (vertical) the user area, maintaining a minimum distance of $1.5~$m. We use the notation (A, B), where A = \{N, O\} represents the scenarios of neglecting (N) or optimizing (O) \gls{RIS} phase shifts in response to temperature changes, and B = \{H, V\} denotes the eavesdropper's horizontal (H) or vertical (V) position relative to the user area. In Figs. \ref{fig: heat map NH} and \ref{fig: heat map NV} neglecting (N) temperature effects results in higher received power in the eavesdropper area (red rectangle). Conversely, optimizing (O) the RIS phase shifts significantly reduces the received power in the eavesdropper area as shown in Figs. \ref{fig: heat map OH} and \ref{fig: heat map OV}.

\begin{figure}[t]
\centering
\begin{subfigure}{0.24\textwidth}
    \caption{Scenario 1, Benchmark (N,H).}
    \includegraphics[width=\textwidth,height=0.7\textwidth]{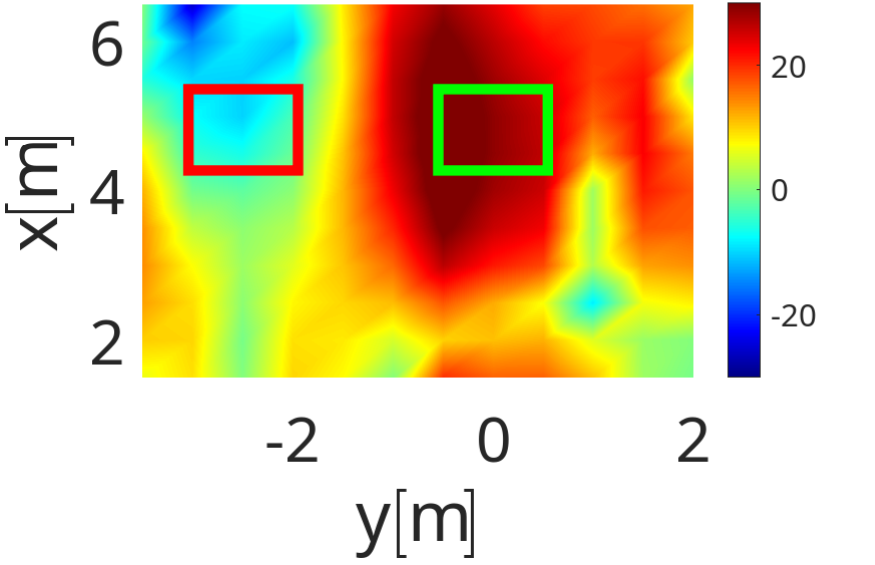}
    \label{fig: heat map NH}
\end{subfigure}
\begin{subfigure}{0.24\textwidth}
    \caption{Scenario 1, Proposed (O,H).}
    \includegraphics[width=\textwidth,height=0.7\textwidth]{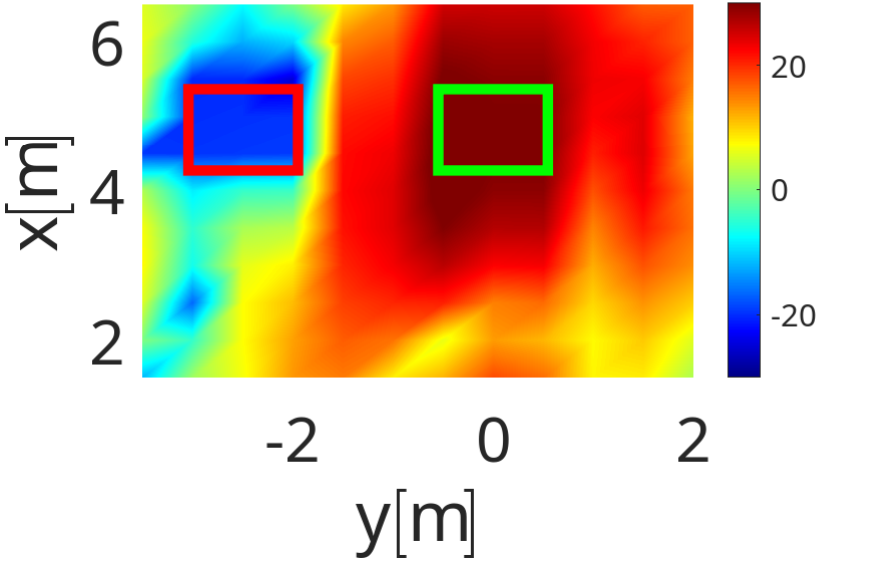}
    \label{fig: heat map OH}
\end{subfigure}
\begin{subfigure}{0.24\textwidth}
    \caption{Scenario 2, Benchmark (N,V).}
   \includegraphics[width=\textwidth,height=0.7\textwidth]{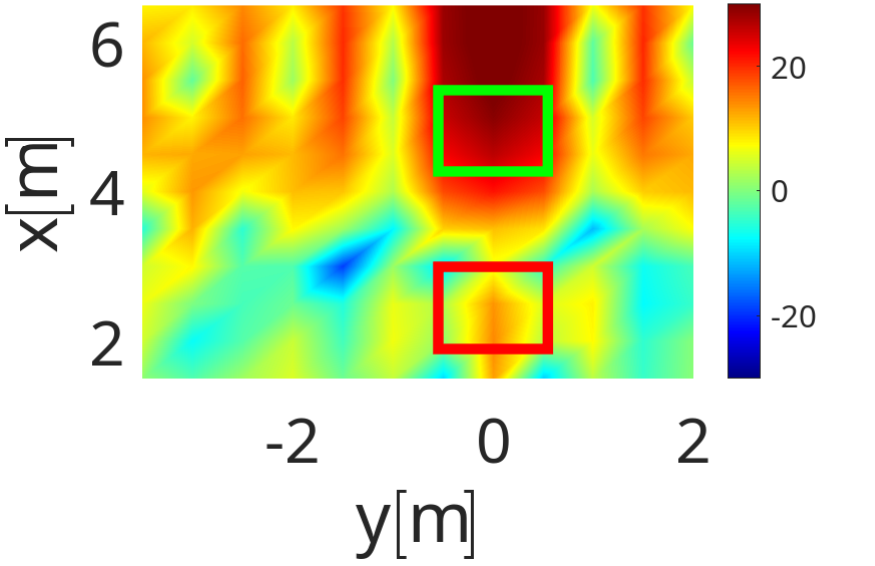}
    \label{fig: heat map NV}
\end{subfigure}
\begin{subfigure}{0.24\textwidth}
    \caption{Scenario 2, Proposed (O,V).}
    \includegraphics[width=\textwidth,height=0.7\textwidth]{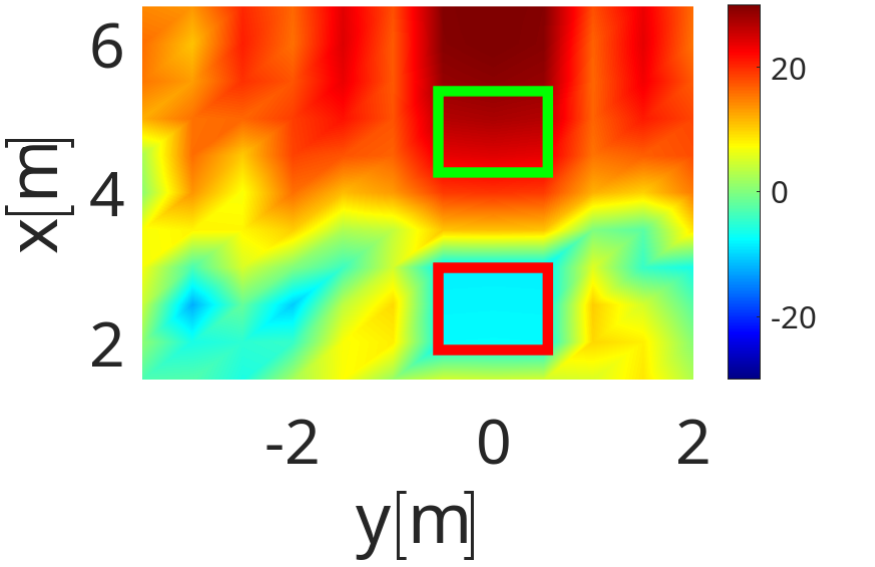}
    \label{fig: heat map OV}
\end{subfigure}

\caption{Received power (dB) for two cases. The first row includes results when the area of the eavesdropper is in a different $\y$ position, and the bottom row demonstrates results for the case eavesdropper's area is at the bottom of the user area.}
\label{fig: heat map}
\vspace{-5mm}
\end{figure}

In Fig. \ref{fig: RS-distance}, we plot the secure rate $\log_2(\gamma)$ as a function of the distance between the eavesdropper and user areas, the secure rate rises up. It shows that the secure rate increases as the distance between the user and the eavesdropper grows, whether the separation is vertical or horizontal. However, at the same distance, the secure rate is higher when the eavesdropper is positioned horizontally (same $x$, different $y$) compared to vertically (same $y$, different $x$). This is because the \gls{RIS} has more elements along the $y$-axis, allowing for better control of the reflected signals and reducing leakage when the eavesdropper is horizontally displaced.

\begin{figure}[t]
    \centering
    \includegraphics[width=0.45\textwidth]{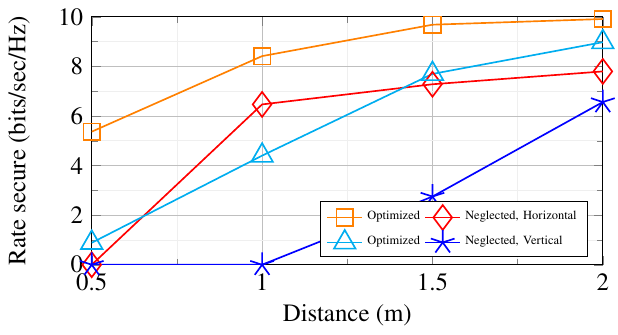}
    \caption{The rate secure (bits/sec/Hz) versus the distance of nearest points in eavesdropper and user areas when the eavesdropper's area is on the left or bottom of the user's area.}
    \label{fig: RS-distance}
    \vspace{-5mm}
\end{figure}

\section{Conclusion}
In this paper, we first have the impact of changing the temperature on the phase shift of the \gls{RIS}. Moreover, we have developed an algorithm to cope with the phase shift range limitation due to the increasing temperature leading to the maximizing secure rate. The simulation results demonstrated the necessity of considering the changing temperature in \gls{LC}-\gls{RIS} phase shift design.

\bibliographystyle{IEEEtran}
\bibliography{References}

\end{document}